\documentclass[journal]{article}

\usepackage{graphicx,url}
\usepackage{dcolumn}
\usepackage{bm}
\usepackage{amsmath,amsfonts,amssymb,citesort}

\newtheorem{definition}{\noindent{\it Definition}}[section]
\newtheorem{theorem}{\noindent{\it Theorem}}[section]

\newtheorem{lemma}[theorem]{\noindent{\it Lemma}}

\newtheorem{remark}[theorem]{\noindent{\it Remark}}

\newenvironment{proof}{\noindent{\it Proof:}}{$\hfill$ $\Box$\\ }

\begin{document}

\title{On Negacyclic MDS-Convolutional Codes}

 \author{Giuliano G. La Guardia
\thanks{Giuliano Gadioli La Guardia is with Department of Mathematics and Statistics,
State University of Ponta Grossa (UEPG), 84030-900, Ponta Grossa,
PR, Brazil. }}

\maketitle

\begin{abstract}
New families of classical and quantum optimal negacyclic
convolutional codes are constructed in this paper. This optimality
is in the sense that they attain the classical (quantum)
generalized Singleton bound. The constructions presented in this
paper are performed algebraically and not by computational search.
\end{abstract}

\section{Introduction}\label{Intro}

Much effort have been paid in order to construct good quantum
error-correcting codes (QECC)
\cite{Calderbank:1998,Steane:1999,Nielsen:2000,Ashikmin:2001,Ketkar:2006,LaGuardia:2009,LaGuardia:2011}
as well as quantum convolutional codes with good parameters
\cite{Ollivier:2003,Ollivier:2004,Aido:2004,Grassl:2005,Grassl:2006,Grassl:2007,Aly:2007,Klapp:2007,Forney:2007,LaGuardia:2013II}.
On the other hand, the investigation of the class of (classical)
convolutional codes and their corresponding properties as well as
constructions of maximum-distance-separable (MDS) convolutional
codes (i.e., codes attaining the generalized Singleton bound
\cite{Rosenthal:1999}) have also been presented in the literature
\cite{Forney:1970,Lee:1976,Piret:1988,York:1999,Rosenthal:1999,Hole:2000,Rosenthal:2001,Hutchinson:2005,Gluesing:2006,Schmale:2006,Luerssen:2006,Luerssen:2008,Climent:2008,Iglesias:2009,LaGuardia:2013I,LaGuardia:2013II}.

In this paper, we utilize the class of negacyclic codes
\cite{Berlekamp:1967,Krishna:1990,Aydin:2001,Blackford:2008,Liu:2012,Bakshi:2012}
in order to construct classical and quantum MDS convolutional
codes. More precisely, we apply the famous method proposed by
Piret \cite{Piret:1988} (generalized recently by Aly \emph{et al.}
\cite{Aly:2007}), which consists in the construction of
(classical) convolutional codes derived from block codes. An
advantage of our techniques of construction lie in the fact that
all new (classical and quantum) convolutional codes are generated
algebraically and not by computational search, in contrast with
many works where only exhaustively computational search or even
specific codes are constructed.

Our classical convolutional MDS codes constructed here have
parameters

\begin{itemize}
\item $(n, n-2i+1, 2 ; 1, 2i+2 {)}_{q^{2}}$, where $q\equiv 1$
(mod $4$) is a power of an odd prime, $n=q^{2}+1$ and $2\leq i
\leq n/2- 1$;

\item $(n, n-2i+2, 2; 1, 2i+1)_{q^{2}}$, where $q$ is a power of
an odd prime, $n=(q^{2}+1)/2$ and $2\leq i\leq (n-1)/2$;

\item $(n, n-2i+1, 2 ; 1, 2i+2 {)}_{q^{2}}$, where $q\geq 5$ is a
power of an odd prime, $n=(q^{2}+1)/2$ and $2\leq i\leq
(n-1)/2-1$.
\end{itemize}

The new convolutional stabilizer MDS codes have parameters

\begin{itemize}
\item ${[(n, n-4i+2, 1; 2, 2i+2)]}_{q}$, where $q\equiv 1$ (mod
$4$) is a power of an odd prime, $n=q^{2}+1$ and $2\leq i \leq
(q-1)/2$;

\item ${[(n, n-4i+4, 1; 2, 2i+1)]}_{q}$, where $q \geq 7$ is a
power of an odd prime, $n=(q^{2}+1)/2$ and $2\leq i \leq (q-1)/2$.
\end{itemize}

We observe that the order between the degree and the memory are
changed when comparing the parameters of classical and quantum
convolutional codes. We adopt this notation to keep the same
notation utilized in \cite{Aly:2007}.

The paper is organized as follows. In Sections~\ref{II} we review
basic concepts on negacyclic codes. In Sections~\ref{III} and
\ref{IV}, we review of concepts concerning classical and quantum
convolutional codes, respectively. In Section~\ref{V}, we propose
constructions of new families of classical MDS convolutional
derived from negacyclic codes. In Section~\ref{VI} we construct
new optimal (MDS) quantum convolutional codes and, in
Section~\ref{VII}, a brief summary of this work is described.

\section{Negacyclic codes}\label{II}

The class of negacyclic codes
\cite{Berlekamp:1967,Krishna:1990,Aydin:2001,Blackford:2008,Liu:2012,Bakshi:2012,Kai:2013A}
have been studied in the literature. This class of codes are a
particular class of a more general class of constacyclic codes
\cite{Blackford:2008}. In this section we review the basic
concepts of these codes.

Throughout this paper, we always assume that $q$ is a power of an
odd prime, ${ \mathbb F}_{q}$ is a finite field with $q$ elements
and $n$ is a positive integer with $\gcd (n, q)=1$. Analogously to
cyclic codes, if we consider the quotient ring $R_{n}={\mathbb
F}_{q}/(x^{n}+1)$, then a negacyclic code is a principal ideal of
$R_{n}$ under the usual correspondence ${\bf c} = (c_0 , c_1 ,
\ldots , c_{n-1})\longrightarrow c_0 + c_1 x + \ldots + c_{n-1}
x^{n-1}$ (mod$(x^{n}+1)$). The generator polynomial $g(x)$ of a
negacyclic code $C$ satisfies $g(x)|(x^{n}+1)$. The roots of
$(x^{n}+1)$ are the roots of $(x^{2n}-1)$ which are not roots of
$(x^{n}-1)$ in some extension field of ${ \mathbb F}_{q^{2}}$
(since we will work with codes endowed with the Hermitian inner
product).

Consider that $m= \ {{ord}_{2n}}(q^{2})$ and let $\beta$ be a
primitive $2n$th root of unity in ${\mathbb F}_{q^{2m}}$ (so
$\alpha = {\beta}^{2} \in {\mathbb F}_{q^{2m}}$ is a primitive
$n$th root of unity). Then the roots of $x^{n} + 1$ are given by
${\beta}^{2i+1}$ $0 \leq i \leq n-1$. Put ${\mathbb O}_{2n}= \{1,
3, \ldots, 2n-1\}$; the defining set of a negacyclic code $C$ of
length $n$ generated by $g(x)$ is given by ${ \mathcal Z} = \{i
\in {\mathbb O}_{2n}| {\beta}^{i} \ \ is \ \ root \ \ of \ \
g(x)\}$. The defining set is a union of $q^{2}$-ary cyclotomic
cosets given by ${\mathcal C}_{i} = \{i, iq^{2}, \ldots, i
q^{2(m_{i}-1)} \}$, where $m_{i}$ is the smallest positive integer
such that $iq^{2(m_{i})}\equiv i$ (mod $2n$). The minimal
polynomial (over ${\mathbb F}_{q^{2}}$) of ${\beta}^{j} \in {
\mathbb F}_{q^{2m}}$ is denoted by ${M}^{(j)}(x)$ and it is given
by ${M}^{(j)}(x)=\displaystyle \prod_{j \in {\mathcal{C}}_{i}}
(x-{\beta}^{j})$. The dimension of $C$ is given by $n - |{
\mathcal Z}|$. The BCH bound for Constacyclic codes (see
\cite{Krishna:1990,Aydin:2001}) asserts that is $C$ is a
$q^{2}$-ary negacyclic code of length $n$ with generator
polynomial $g(x)$ and if $g(x)$ has the elements $\{
{\beta}^{2i+1}| 0 \leq i\leq  d-2 \}$ as roots, where $\beta$ is a
primitive $2n$th root of unity, then the minimum distance $d_{C}$
of $C$ satisfies $d_{C}\geq d$.

\section{Classical Convolutional Codes}\label{III}

The class of (classical) convolutional codes is a well-studied
class of codes
\cite{Forney:1970,Piret:1988,Johannesson:1999,Huffman:2003,Aly:2007,Klapp:2007}.
We assume the reader is familiar with the theory of convolutional
codes. For more details, see \cite{Johannesson:1999}. Recall that
a polynomial encoder matrix $G(D)=(g_{ij}) \in { \mathbb
F}_{q}{[D]}^{k \times n}$ is called \emph{basic} if $G(D)$ has a
polynomial right inverse. A basic generator matrix is called
\emph{reduced} (or minimal
\cite{Rosenthal:2001,Huffman:2003,Luerssen:2008}) if the overall
constraint length $\gamma =\displaystyle\sum_{i=1}^{k}
{\gamma}_i$, where ${\gamma}_i = {\max}_{1\leq j \leq n} \{ \deg
g_{ij} \}$, has the smallest value among all basic generator
matrices (in this case the overall constraint length $\gamma$ will
be called the \emph{degree} of the resulting code).

\begin{definition}\cite{Klapp:2007}
A rate $k/n$ \emph{convolutional code} $C$ with parameters $(n, k,
\gamma ; \mu,$ $d_{f} {)}_{q}$ is a submodule of ${ \mathbb F}_q
{[D]}^{n}$ generated by a reduced basic matrix $G(D)=(g_{ij}) \in
{ \mathbb F}_q {[D]}^{k \times n}$, that is, $C = \{ {\bf
u}(D)G(D) | {\bf u}(D)\in { \mathbb F}_{q} {[D]}^{k} \}$, where
$n$ is the length, $k$ is the dimension, $\gamma
=\displaystyle\sum_{i=1}^{k} {\gamma}_i$ is the \emph{degree},
$\mu = {\max}_{1\leq i\leq k}\{{\gamma}_i\}$ is the \emph{memory}
and $d_{f}=$wt$(C)=\min \{wt({\bf v}(D)) \mid {\bf v}(D) \in C,
{\bf v}(D)\neq 0 \}$ is the \emph{free distance} of the code.
\end{definition}

The Hermitian inner product is defined as $\langle {\bf u}(D)\mid
{\bf v}(D){\rangle}_{h} = {\sum}_i {\bf u}_i \cdot {\bf v}_i^{q}$,
where ${\bf u}_i, {\bf v}_i\in { \mathbb F}_{q^{2}}^{n}$ and ${\bf
v}_i^{q} = (v_{1i}^{q}, \ldots, v_{ni}^{q})$. The Hermitian dual
of the code $C$ is defined by $C^{{\perp}_{h} }=\{ {\bf u}(D) \in
{ \mathbb F}_{q^{2}} {[D]}^{n}\mid \langle {\bf u}(D)\mid {\bf
v}(D){\rangle}_{h} = 0$ for all ${\bf v}(D)\in C\}$.

Let $C$ an ${[n, k, d]}_{q}$ block code with parity check matrix
$H$. We split $H$ into $\mu+1$ disjoint submatrices $H_i$ such
that $H = \left[
\begin{array}{c}
H_0\\
H_1\\
\vdots\\
H_{\mu}\\
\end{array}
\right]$, where each $H_i$ has $n$ columns, obtaining the
polynomial matrix $G(D) =  {\tilde H}_0 + {\tilde H}_1 D + {\tilde
H}_2 D^2 + \ldots + {\tilde H}_{\mu} D^{\mu}$, where the matrices
${\tilde H}_i$, for all $1\leq i\leq \mu$, are derived from the
respective matrices $H_i$ by adding zero-rows at the bottom in
such a way that the matrix ${\tilde H}_i$ has $\kappa$ rows in
total, where $\kappa$ is the maximal number of rows among the
matrices $H_i$. As it is well known, the matrix $G(D)$ generates a
convolutional code. Note that $\mu$ is the memory of the resulting
convolutional code generated by the matrix $G(D)$.

\begin{theorem}\cite[Theorem 3]{Aly:2007}\label{A}
Let $C \subseteq { \mathbb F}_q^n$ be a linear code with
parameters ${[n, k, d]}_{q}$ and assume also that $H \in { \mathbb
F}_q^{(n-k)\times n}$ is a parity check matrix for $C$ partitioned
into submatrices $H_0, H_1, \ldots, H_{\mu}$ as above such that
$\kappa =$ rk$H_0$ and rk$H_i \leq \kappa$
for $1 \leq i\leq \mu$ and consider the polynomial matrix $G(D)$ as given above. Then we have:\\
(a) The matrix $G(D)$ is a reduced basic generator matrix;\\
(b) If $C^{\perp}\subset C$ (resp. ${C}^{{\perp}_{h}}\subset C$),
then the convolutional code $V = \{ {\bf v}(D) = {\bf u}(D)G(D)
\mid {\bf u}(D)\in { \mathbb F}_q^{n-k} [D] \}$ satisfies $V \subset V^{\perp}$
(resp. $V \subset {V}^{{\perp}_{h}}$);\\
(c) If $d_f$ and $d_f^{\perp}$ denote the free distances of $V$
and $V^{\perp}$, respectively, $d_i$ denote the minimum distance
of the code $C_i = \{ {\bf v}\in { \mathbb F}_q^n \mid {\bf v}
{\tilde H}_i^t =0 \}$ and $d^{\perp}$ is the minimum distance of
$C^{\perp}$, then one has $\min \{ d_0 + d_{\mu} , d \} \leq
d_f^{\perp} \leq  d$ and $d_f \geq d^{\perp}$.
\end{theorem}

Recall that the (classical) generalized Singleton bound
\cite[Theorem 2.2]{Rosenthal:1999} of an $(n, k, \gamma ; \mu,
d_{f} {)}_{q}$ convolutional code is given by
\begin{eqnarray}
d_{f}\leq (n-k)[ \lfloor \gamma/k \rfloor + 1 ] + \gamma +1
\end{eqnarray}
If the parameters of a convolutional code $C$ satisfies (1) with
equality then $C$ is said maximum-distance-separable (MDS).


\section{Quantum Convolutional Codes}\label{IV}

A quantum convolutional code is defined by means of its
stabilizer, which is a subgroup of the infinite version of the
Pauli group, consisting of tensor products of generalized Pauli
matrices acting on a semi-infinite stream of qudits. The
stabilizer can be defined by a stabilizer matrix of the form
$$S(D) = ( X(D)\mid Z(D)) \in { \mathbb F}_{q}{[D]}^{(n-k)\times 2n}$$
satisfying $X(D){Z(1/D)}^{t} - Z(D){X(1/D)}^{t}=0$ (symplectic
orthogonality). More precisely, consider a quantum convolutional
code $C$ defined by a full-rank stabilizer matrix $S(D)$ given
above. Then $C$ is a rate $k/n$ code with parameters $[(n, k, \mu;
\gamma, d_{f} ){]}_{q}$, where $n$ is the frame size, $k$ is the
number of logical qudits per frame, $\mu = {\max}_{1\leq i\leq
n-k, 1\leq j\leq n} \{ \max \{ \deg {X}_{ij}(D), \deg {Z}_{ij}
(D)\} \}$ is the memory, $d_{f}$ is the free distance and $\gamma$
is the degree of the code. Similarly as in the classical case, the
constraint lengths are defined as ${\gamma}_{i} = {\max}_{1\leq
j\leq n}$ $\{ \max \{\deg X_{ij}(D), \deg Z_{ij}(D)\} \}$, and the
overall constraint length is defined as $\gamma
=\displaystyle\sum_{i=1}^{n-k} {\gamma}_{i}$.

Next, let ${\mathbb H} = {\mathbb C}^{q^n} = {\mathbb C}^{q}
\otimes \ldots \otimes {\mathbb C}^{q}$ be the Hilbert space and
$\mid$$x \rangle$ be the vectors of an orthonormal basis of
${\mathbb C}^{q}$, where the labels $x$ are elements of ${ \mathbb
F}_{q}$. Consider $a, b \in { \mathbb F}_{q}$ and take the unitary
operators $X(a)$ and $Z(b)$ in ${\mathbb C}^{q}$ defined by
$X(a)$$\mid$$x \rangle =$$\mid$$x + a\rangle$ and $Z(b)$$\mid$$x
\rangle = w^{tr(bx)}$$\mid$$x\rangle$, respectively, where $w=\exp
(2\pi i/ p)$ is a primitive $p$-th root of unity, $p$ is the
characteristic of ${ \mathbb F}_{q}$ and $tr$ is the trace map
from ${ \mathbb F}_{q}$ to ${ \mathbb F}_{p}$. Considering the
\emph{error basis} ${\mathbb E} = \{X(a), Z(b) | a, b \in {
\mathbb F}_{q} \}$, one defines the set $P_{\infty}$ (according to
\cite{Klapp:2007}) as the set of all infinite tensor products of
matrices $N\in \langle M\mid M \in {\mathbb E} \rangle$, in which
all but finitely many tensor components are equal to $I$, where
$I$ is the $q\times q$ identity matrix. Then one defines the
\emph{weight} wt of $A\in P_{\infty}$ as its (finite) number of
nonidentity tensor components. In this context, one says that a
quantum convolutional code has free distance $d_{f}$ if and only
if it can detect all errors of weight less than $d_{f}$, but
cannot detect some error of weight $d_{f}$. The code $C$ is
\emph{pure} if does not exist errors of weight less than $d_{f}$
in the stabilizer of $C$.


\section{The New Convolutional MDS Codes}\label{V}

In this section we propose the construction of new classical
convolutional codes. In order to proceed further, let us recall
some results shown in the literature:

\begin{lemma}\label{kai1}\cite[Lemma 4.1]{Kai:2013}
Let $n=q^{2}+1$, where $q\equiv 1$ (mod $4$) is a power of an odd
prime and suppose that $s=n/2$. Then the $q^{2}$-ary cosets modulo
$2n$ are given by: ${\mathcal C}_{s}=\{ s\}$, ${\mathcal
C}_{3s}=\{ 3s\}$ and ${\mathcal C}_{s-2i}=\{ s-2i, s+2i\}$, where
$1\leq i\leq s-1$.
\end{lemma}

\begin{lemma}\label{kai2}\cite[Lemma 4.4]{Kai:2013}
Let $n=(q^{2}+1)/2$, where $q$ is a power of an odd prime. Then
the $q^{2}$-ary cosets modulo $2n$ containing all odd integers
from $1$ to $2n-1$ are given by: ${\mathcal C}_{n}=\{ n\}$, and
${\mathcal C}_{2i-1}=\{ 2i-1, 1-2i\}$, where $1\leq i\leq
(n-1)/2$.
\end{lemma}

Recall the concept of negacyclic BCH codes:

\begin{definition}(Negacyclic BCH codes)
Let $q$ be a power of an odd prime with $\gcd (n, q)=1$. Let
$\beta$ be a primitive $2n$th root of unity in ${\mathbb
F}_{q^{m}}$. A negacyclic code $C$ of length $n$ over ${ \mathbb
F}_{q}$ is a BCH code with designed distance $\delta$ if, for some
odd integer $b\geq 1,$ we have
$$g(x)= \operatorname{lcm} \{{M}^{(b)}(x), {M}^{(b+2)}(x), \ldots,
{M}^{[b+2(\delta-2)]}(x)\},$$ i.e., $g(x)$ is the monic polynomial
of smallest degree over ${ \mathbb F}_{q}$ having ${{\alpha}^{b}},
{{\alpha}^{b+2}}\ldots,$ ${{\alpha}^{[b+2(\delta-2)]}}$ as zeros.
Therefore, $c\in C$ if and only if
$c({\alpha}^{b})=c({{\alpha}^{(b+2)}})=\ldots =
c({{\alpha}^{[b+2(\delta-2)]}})=0$. Thus the code has a string of
$\delta - 1$ consecutive odd powers of $\beta$ as zeros.
\end{definition}

\begin{remark}\label{rem1}
Let ${\mathcal B} =\{ b_{1}, \ldots, b_{l}\}$ be a basis of ${
\mathbb F}_{q^{l}}$ over ${ \mathbb F}_{q}$. If $u = (u_1,\ldots
,u_{n}) \in { \mathbb F}_{q^{l}}^{n}$ then one can write the
vectors $u_{i}$, $1\leq i\leq n$, as linear combinations of the
elements of ${\mathcal B}$, that is, $u_{i} = u_{i1}b_{1} +\ldots
+ u_{il}b_{l}$. Consider that $u^{(j)} = (u_{1j},\ldots, u_{nj})$
are vectors in ${ \mathbb F}_{q}^{n}$ with $1\leq j\leq l$. Then,
if $v \in { \mathbb F}_{q}^{n}$, one has $v\cdot u=0$ if and only
if $v \cdot u^{(j)} = 0$ for all $1\leq j\leq l$.
\end{remark}

In the following theorem we construct a parity-check matrix for
negacyclic codes:

\begin{theorem}\label{paritynega}
Assume that $q$ is a power of an odd prime, $\gcd (n, q)=1$, and
$m= \ {{ord}_{2n}}(q)$. Let $\beta$ be a primitive $2n$th root of
unity in ${\mathbb F}_{q^{m}}$. Let $b$ be an odd positive integer
with $1 \leq b\leq 2n-1$. Then a parity-check matrix for the BCH
negacyclic code $C$ of length $n$ and designed distance $\delta$,
generated by the polynomial $g(x)= \operatorname{lcm}
\{{M}^{(b)}(x), {M}^{(b+2)}(x), \ldots,
{M}^{[b+2(\delta-2)]}(x)\}$, is the matrix
\begin{eqnarray*}
H_{\delta , b} = \\=  \left[
\begin{array}{ccccc}
1 & {{\beta}^{b}} & {{\beta}^{2b}} & \cdots & {{\beta}^{(n-1)b}} \\
1 & {{\beta}^{(b+2)}} & {{\beta}^{2(b+2)}} & \cdots & {{\beta}^{(n-1)(b+2)}}\\
1 & {{\beta}^{(b+4)}} & {{\beta}^{2(b+4)}} & \cdots & {{\beta}^{(n-1)(b+4)}}\\
\vdots & \vdots & \vdots & \vdots & \vdots\\
1 & {\beta}^{[b+2(\delta-2)]} & {\beta}^{2[b+2(\delta-2)]} & \cdots & {\beta}^{(n-1)[b+2(\delta-2)]}\\
\end{array}
\right],
\end{eqnarray*}
where each entry is replaced by the corresponding column of $m$
elements from ${ \mathbb F}_{q}$ and then removing any linearly
dependent rows.
\end{theorem}

\begin{proof}
Assume that ${\bf c} =(c_0, c_1, \ldots , c_{n-1}) \in C$. Thus we
have ${\bf c}({{\beta}^{b}})={\bf c}({{\beta}^{b+2}})={\bf
c}({{\beta}^{b+4}})= \ldots ={\bf
c}({{\beta}^{[b+2(\delta-2)]}})=0$, hence

\begin{eqnarray*}
\left[
\begin{array}{ccccc}
1 & {{\beta}^{b}} & {{\beta}^{2b}} & \cdots & {{\beta}^{(n-1)b}} \\
1 & {{\beta}^{(b+2)}} & {{\beta}^{2(b+2)}} & \cdots & {{\beta}^{(n-1)(b+2)}}\\
1 & {{\beta}^{(b+4)}} & {{\beta}^{2(b+4)}} & \cdots & {{\beta}^{(n-1)(b+4)}}\\
\vdots & \vdots & \vdots & \vdots & \vdots\\
1 & {\beta}^{[b+2(\delta-2)]} & {\beta}^{2[b+2(\delta-2)]} & \cdots & {\beta}^{(n-1)[b+2(\delta-2)]}\\
\end{array}
\right]\cdot \left[
\begin{array}{c}
c_0\\
c_1\\
c_2\\
\vdots\\
c_{n-1}\\
\end{array}
\right]= \left[
\begin{array}{c}
0\\
0\\
\vdots\\
0\\
\end{array}
\right]_{({\delta}-1, 1)}.
\end{eqnarray*}
From Remark~\ref{rem1} and from the definition of BCH negacyclic
codes, the result follows.
\end{proof}

Now we are ready to show one of the main results of this section:

\begin{theorem}\label{mainI}
Let $n=q^{2}+1$, where $q\equiv 1$ mod $4$ is a power of an odd
prime and suppose that $s=n/2$. Then there exist MDS convolutional
codes with parameters $(n, n-2i+1, 2 ; 1, 2i+2 {)}_{q^{2}}$, where
$2\leq i\leq n/2-1$.
\end{theorem}

\begin{proof}
First, note that $ \gcd (n, q) = 1$ and ${{ord}_{2n}}(q^{2})=2$.
Let $\beta$ be a primitive $2n$th root of unity in ${\mathbb
F}_{q^{2m}}$. Consider that $C_2$ is the negacyclic BCH code of
length $n$ over ${ \mathbb F}_{q^{2}}$ generated by the product of
the minimal polynomials
\begin{eqnarray*}
C_2 = \langle g_{2} (x)  \rangle= \langle {M}^{(s)}(x)
{M}^{(s+2)}(x) \cdot \ldots \cdot {M}^{(s+2i)}(x) \rangle,
\end{eqnarray*}
where  $2\leq i\leq s-1$.

By Theorem~\ref{paritynega}, a parity check matrix of $C_2$ is
obtained from the matrix
\begin{eqnarray*}
H_{2} = \left[
\begin{array}{ccccc}
1 & {{\beta}^{s}} & {{\beta}^{2s}} & \cdots & {{\beta}^{(n-1)s}} \\
1 & {{\beta}^{(s+2)}} & {{\beta}^{2(s+2)}} & \cdots & {{\beta}^{(n-1)(s+2)}}\\
1 & {{\beta}^{(s+4)}} & {{\beta}^{2(s+4)}} & \cdots & {{\beta}^{(n-1)(s+4)}}\\
\vdots & \vdots & \vdots & \vdots & \vdots\\
1 & {\beta}^{(s+2i)} & {{\beta}^{2(s+2i)}} & \cdots & {\beta}^{(n-1)(s+2i)}\\
\end{array}
\right]
\end{eqnarray*}
by expanding each entry as a column vector (containing $2$ rows)
with respect to some ${ \mathbb F}_{q^{2}}-$basis $\beta$ of ${
\mathbb F}_{q^4}$ and then removing one linearly dependent row.
From Lemma~\ref{kai1}, this new matrix $H_{C_2}$ has rank $2i+1$,
so $C_{2}$ has dimension $n - 2i-1$. From the BCH bound for
negacyclic codes it follows that the minimum distance $d_2$ of
$C_{2}$ satisfies $d_2 \geq 2i+2$. Thus, from the (classical)
Singleton bound, one concludes that $C_2$ is a MDS code with
parameters ${[n, n-2i-1, 2i+2]}_{q^{2}}$ and, consequently, its
Hermitian dual code has dimension $2i+1$.

Next we assume that $C_1$ is the negacyclic BCH code of length $n$
over ${ \mathbb F}_{q^{2}}$ generated by the product of the
minimal polynomials
\begin{eqnarray*}
C_1 =\langle g_1 (x)\rangle =\langle {M}^{(s)}(x) {M}^{(s+2)}(x)
\cdot \ldots \cdot {M}^{[s+2(i-1)]}(x) \rangle,
\end{eqnarray*}

Similarly, by Theorem~\ref{paritynega}, $C_1$ has a parity check
matrix derived from the matrix
\begin{eqnarray*}
H_{1} = \left[
\begin{array}{ccccc}
1 & {{\beta}^{s}} & {{\beta}^{2s}} & \cdots & {{\beta}^{(n-1)s}} \\
1 & {{\beta}^{(s+2)}} & {{\beta}^{2(s+2)}} & \cdots & {{\beta}^{(n-1)(s+2)}}\\
1 & {{\beta}^{(s+4)}} & {{\beta}^{2(s+4)}} & \cdots & {{\beta}^{(n-1)(s+4)}}\\
\vdots & \vdots & \vdots & \vdots & \vdots\\
1 & {\beta}^{[s+2(i-1)]} & {\beta}^{2[s+2(i-1)]} & \cdots & {\beta}^{(n-1)[s+2(i-1)]}\\
\end{array}
\right]
\end{eqnarray*}
by expanding each entry as a column vector with respect to some ${
\mathbb F}_{q^{2}}-$basis $\beta$ of ${ \mathbb F}_{q^4}$ (already
done, since $H_1$ is a submatrix of $H_{2}$) and then removing one
linearly dependent row. From Lemma~\ref{kai1}, this new matrix
$H_{C_1}$ has rank $2i-1$, so $C_{1}$ has dimension $n - 2i+1$.
From the BCH bound for negacyclic codes, the minimum distance
$d_1$ of $C_1$ satisfies $d_1 \geq 2i$, so $C_1$ is an ${[n,
n-2i+1, 2i]}_{q^{2}}$ MDS code. Thus, its Hermitian dual code has
dimension $2i-1$.

Now, let $C_0$ be the negacyclic BCH code of length $n$ over ${
\mathbb F}_{q^{2}}$ generated by the minimal polynomial
${M}^{(s+2i)}(x)$. Then $C_0$ has parameters ${[n, n-2, d_0 \geq
2]}_{q^{2}}$. A parity check matrix $H_{C_0}$ of $C_0$ is given by
expanding the entries of the matrix
\begin{eqnarray*}
H_0 = \left[
\begin{array}{ccccc}
1 & {{\alpha}^{(s+2i)}} & {{\alpha}^{2(s+2i)}} & \cdots & {{\alpha}^{(n-1)(s+2i)}} \\
\end{array}
\right]
\end{eqnarray*}
with respect to $\beta$ (already done, since $H_0$ is a submatrix
of $H_{2}$).

Further, let us construct the convolutional code $V$ generated by
the reduced basic (according to Theorem~\ref{A} Item (a))
generator matrix
\begin{eqnarray*}
G(D)=\tilde H_{C_1}+ \tilde H_{C_0} D,
\end{eqnarray*}
where $\tilde H_{C_1} = H_{C_1}$ and $\tilde H_{C_0}$ is obtained
from $H_{C_0}$ by adding zero-rows at the bottom such that $\tilde
H_{C_0}$ has the number of rows of $H_{C_1}$ in total. By
construction, $V$ is a unit-memory convolutional code of dimension
$2i-1$ and degree ${\delta}_{V} = 2$. We know that the Hermitian
dual $V^{{\perp}_{h}}$ of $V$ has dimension $n-2i+1$ and degree
$2$. By Theorem~\ref{A} Item (c), the free distance of
$V^{{\perp}_{h}}$ is bounded by $\min \{ d_0 + d_1 , d_2 \} \leq
d_{f}^{{\perp}_{h}} \leq d_2$, where $d_i$ is the minimum distance
of the code $C_i = \{ {\bf v}\in { \mathbb F}_q^n \mid {\bf v}
{\tilde H}_{C_i}^t =0 \}$. From construction one has $d_2 = 2i+2$,
$d_1 = 2i$ and $d_0 \geq 2$, so $V^{{\perp}_{h}}$ has parameters
$(n, n-2i+1, 2; 1, 2i+2)_{q^{2}}$. It is easy to see that the
parameters of $V^{{\perp}_{h}}$ satisfies (1) with equality, so
$V^{{\perp}_{h}}$ is MDS.
\end{proof}

Theorem~\ref{mainII} given in the following is the second main
result of this section:

\begin{theorem}\label{mainII}
Let $n=(q^{2}+1)/2$, where $q$ is a power of an odd prime. Then
there exist MDS convolutional codes with parameters $(n, n-2i+2,
2; 1, 2i+1)_{q^{2}}$, where $2\leq i\leq (n-1)/2$.
\end{theorem}

\begin{proof}
It suffices to consider $C_2$ be the code generated by $\langle
{M}^{(1)}(x) {M}^{(3)}(x) \cdot \ldots \cdot {M}^{(2i-1)}(x)
\rangle$, where  $2\leq i\leq (n-1)/2$, $C_1$ be the negacyclic
BCH code generated by $\langle g_1 (x)\rangle =\langle
{M}^{(1)}(x) {M}^{(3}(x) \cdot \ldots \cdot {M}^{(2i-3)}(x)
\rangle$ and $C_0$ be the negacyclic BCH code generated by
${M}^{(2i-1)}(x)$. Proceeding similarly as in the proof of
Theorem~\ref{mainI}, the result follows.
\end{proof}

\begin{theorem}\label{mainIII}
Let $n=(q^{2}+1)/2$, where $q\geq 5$ is a power of an odd prime.
Then there exist MDS convolutional codes with parameters $(n,
n-2i+1, 2 ; 1, 2i+2 {)}_{q^{2}}$, where $2\leq i\leq
\frac{(n-1)}{2}-1$.
\end{theorem}

\begin{proof}
Consider that $C_2$, $C_1$ and $C_0$ are negacyclic BCH codes of
length $n$ over ${ \mathbb F}_{q^{2}}$ generated, respectively by
$\langle g_{2} (x)  \rangle=\langle {M}^{(n)}(x) {M}^{(n+2)}(x)
\cdot \ldots \cdot {M}^{(n+2i)}(x) \rangle,$ $\langle g_1
(x)\rangle =\langle {M}^{(n)}(x) {M}^{(n+2}(x) \cdot \ldots \cdot
{M}^{(n+2i-2)}(x) \rangle,$ and $\langle g_0 (x)\rangle
=\langle{M}^{(n+2i)}(x)\rangle$. Applying the same procedure given
in the proofs of Theorems~\ref{mainI} and \ref{mainII}, the result
follows.
\end{proof}

\section{New Quantum MDS-Convolutional codes}\label{VI}

As in the classical case, the construction of MDS quantum
convolutional codes is a difficult task. This task is performed in
\cite{Grassl:2005,Grassl:2007,Klapp:2007,Forney:2007} but only in
\cite{Grassl:2005,Klapp:2007} the constructions are made
algebraically. Here, we propose the construction of MDS
convolutional stabilizer codes derived from the convolutional
codes constructed in Section~\ref{V}. To proceed further, let us
recall some results available in the literature:

\begin{lemma}\cite[Proposition 2]{Aly:2007}\label{BB}
Let $C$ be an ${(n, (n - k)/2, \gamma; \mu)}_{q^2}$ convolutional
code such that $C \subseteq {C}^{{\perp}_{h}}$. Then there exists
an ${[(n, k, \mu; \gamma, d_f)]}_{q}$ convolutional stabilizer
code, where $d_f=$ wt$({C}^{{\perp}_{h}} \backslash C)$.
\end{lemma}

\begin{theorem}\label{SingC}\cite{Klapp:2007}
(Quantum Singleton bound) The free distance of an $[(n, k, \mu;
\gamma,$ $d_{f}){]}_{q}$ ${ \mathbb F}_{q^{2}}$-linear pure
convolutional stabilizer code is bounded by
\begin{eqnarray*}
d_{f} \leq \frac{n-k}{2} \left(\left\lfloor
\frac{2\gamma}{n+k}\right \rfloor + 1\right) + \gamma + 1.
\end{eqnarray*}
\end{theorem}

\begin{lemma}\cite{Kai:2013}\label{negaself1}
Let $n=q^{2}+1$, where $q\equiv 1$ (mod $4$) is a power of an odd
prime and suppose that $s=n/2$. If $C$ is a $q^{2}$-ary negacyclic
code of length $n$ with defining set ${ \mathcal
Z}=\displaystyle\cup_{i=0}^{\delta} {\mathcal C}_{s-2i}$, where
$0\leq \delta \leq (q-1)/2$, then $C^{{\perp}_{h}}\subseteq C$.
\end{lemma}

\begin{lemma}\cite{Kai:2013}\label{negaself2}
Let $n=(q^{2}+1)/2$, where $q$ is a power of an odd prime. If $C$
is a $q^{2}$-ary negacyclic code of length $n$ with defining set
${ \mathcal Z}=\displaystyle\cup_{i=1}^{\delta} {\mathcal
C}_{2i-1}$, where $1\leq \delta \leq (q-1)/2$, then
$C^{{\perp}_{h}}\subseteq C$.
\end{lemma}

Now, we are able to show the following two results, in which new
families of quantum convolutional MDS codes are constructed:

\begin{theorem}\label{mainIV}
Let $n=q^{2}+1$, where $q\equiv 1$ (mod $4$) is a power of an odd
prime and suppose that $s=n/2$. Then there exist quantum MDS
convolutional codes with parameters ${[(n, n-4i+2, 1; 2,
2i+2)]}_{q}$, where $2\leq i \leq (q-1)/2$.
\end{theorem}

\begin{proof}
We consider the same notation utilized in Theorem~\ref{mainI}.
From Theorem~\ref{mainI}, there exists a classical convolutional
MDS code $V^{{\perp}_{h}}$ with parameters $(n, n-2i+1, 2; 1,
2i+2)_{q^2}$, for each $2\leq i\leq n/2-1$. This code is the
Hermitian dual of the code $V$ with parameters $(n, 2i-1, 2; 1,
d_f)_{q^2}$. From Lemma~\ref{negaself1} and from Theorem~\ref{A}
Item (b), one has $V \subset V^{{\perp}_{h}}$. Applying
Lemma~\ref{BB}, there exists an ${[(n, n-4i+2, 1; 2, d_f\geq
2i+2)]}_{q}$ convolutional stabilizer code ${ \mathcal Q}$, for
each $2\leq i \leq (q-1)/2$. Replacing the parameters of ${
\mathcal Q}$ in Theorem~\ref{SingC}, the result follows.
\end{proof}

\begin{theorem}\label{mainV}
Let $n=(q^{2}+1)/2$, where $q\geq 7$ is a power of an odd prime.
Then there exist quantum MDS convolutional codes with parameters
$[(n, n-4i+4, 1; 2,$ $2i+1){]}_{q}$, where $2\leq i \leq (q-1)/2$.
\end{theorem}

\begin{proof}
From Theorem~\ref{mainII}, there exists a classical convolutional
MDS code $V^{{\perp}_{h}}$ with parameters $(n, n-2i+2, 2; 1,
2i+1)_{q^2}$, for each $2\leq i\leq (n-1)/2$. This code is the
Hermitian dual of the code $V$ with parameters $(n, 2i-2, 2; 1,
d_f)_{q^2}$. From Lemma~\ref{negaself2} and from Theorem~\ref{A}
Item (b), one has $V \subset V^{{\perp}_{h}}$. Applying
Lemma~\ref{BB}, there exists a convolutional stabilizer code ${
\mathcal Q}$ with parameters ${[(n, n-4i+4, 1; 2, d_f\geq
2i+1)]}_{q}$, for each $2\leq i \leq (q-1)/2$. Replacing the
parameters of ${ \mathcal Q}$ in Theorem~\ref{SingC}, the result
follows.
\end{proof}

In the following we present Tables~\ref{table1} and \ref{table2},
containing the parameters of some new convolutional codes and some
new quantum convolutional codes, respectively, constructed in this
paper. Recall the these codes are optimal in the sense the they
attain the classical (quantum) generalized Singleton bound.

\begin{table}[!hpt]
\begin{center}
\caption{Classical MDS \label{table1}}
\begin{tabular}{|c |}
\hline  New convolutional codes\\
\hline $(n, n-2i+1, 2; 1, 2i+2{)}_{q^{2}}$, $q\equiv 1$( mod $4$), $n=q^{2}+1$, $2\leq i\leq n/2 -1$\\
\hline ${(26, 23, 2; 1, 6)}_{25}$\\
\hline ${(26, 21, 2; 1, 8)}_{25}$\\
\hline ${(26, 19, 2; 1, 10)}_{25}$\\
\hline ${(26, 9, 2; 1, 20)}_{25}$\\
\hline ${(26, 7, 2; 1, 22)}_{25}$\\
\hline ${(26, 5, 2; 1, 24)}_{25}$\\
\hline ${(82, 63, 2; 1, 22)}_{81}$\\
\hline ${(82, 53, 2; 1, 32)}_{81}$\\
\hline ${(82, 43, 2; 1, 42)}_{81}$\\
\hline ${(82, 23, 2; 1, 62)}_{81}$\\
\hline ${(82, 13, 2; 1, 72)}_{81}$\\
\hline
\hline $(n, n-2i+2, 2; 1, 2i+1{)}_{q^{2}}$, $n=(q^{2}+1)/2$, $2\leq i\leq (n-1)/2$\\
\hline ${(5, 3, 2; 1, 5)}_{9}$\\
\hline ${(25, 23, 2; 1, 5)}_{49}$\\
\hline ${(25, 21, 2; 1, 7)}_{49}$\\
\hline ${(25, 19, 2; 1, 9)}_{49}$\\
\hline ${(25, 17, 2; 1, 11)}_{49}$\\
\hline ${(25, 15, 2; 1, 13)}_{49}$\\
\hline ${(25, 13, 2; 1, 15)}_{49}$\\
\hline ${(25, 11, 2; 1, 17)}_{49}$\\
\hline ${(25, 7, 2; 1, 21)}_{49}$\\
\hline
\hline $(n, n-2i+1, 2 ; 1, 2i+2 {)}_{q^{2}}$, $q\geq 5$, $n=(q^{2}+1)/2$ $2\leq i\leq (n-1)/2-1$\\
\hline ${(13, 10, 2; 1, 6)}_{25}$\\
\hline ${(13, 8, 2; 1, 8)}_{25}$\\
\hline ${(13, 6, 2; 1, 10)}_{25}$\\
\hline ${(13, 4, 2; 1, 12)}_{25}$\\
\hline ${(25, 16, 2; 1, 12)}_{49}$\\
\hline ${(25, 10, 2; 1, 18)}_{49}$\\
\hline ${(25, 4, 2; 1, 24)}_{49}$\\
\hline ${(41, 38, 2; 1, 6)}_{81}$\\
\hline ${(41, 32, 2; 1, 12)}_{81}$\\
\hline ${(41, 24, 2; 1, 20)}_{81}$\\
\hline ${(41, 4, 2; 1, 40)}_{81}$\\
\hline ${(61, 32, 2; 1, 32)}_{121}$\\
\hline ${(61, 22, 2; 1, 42)}_{121}$\\
\hline ${(61, 4, 2; 1, 60)}_{121}$\\
\hline
\end{tabular}
\end{center}
\end{table}

\begin{table}[!hpt]
\begin{center}
\caption{Quantum MDS \label{table2}}
\begin{tabular}{|c |}
\hline  New convolutional stabilizer codes\\
\hline $[(n, n-4i+2, 1; 2, 2i+2{)]}_{q}$, $q\equiv 1$( mod $4$), $n=q^{2}+1$, $2\leq i\leq (q-1)/2$\\
\hline
\hline ${[(26, 20, 2; 1, 6)]}_{5}$\\
\hline ${[(82, 80, 2; 1, 4)]}_{9}$\\
\hline ${[(82, 76, 2; 1, 6)]}_{9}$\\
\hline ${[(82, 72, 2; 1, 8)]}_{9}$\\
\hline ${[(82, 68, 2; 1, 10)]}_{9}$\\
\hline ${[(170, 168, 2; 1, 4)]}_{13}$\\
\hline ${[(170, 164, 2; 1, 6)]}_{13}$\\
\hline ${[(170, 160, 2; 1, 8)]}_{13}$\\
\hline ${[(170, 156, 2; 1, 10)]}_{13}$\\
\hline ${[(170, 152, 2; 1, 12)]}_{13}$\\
\hline ${[(170, 148, 2; 1, 14)]}_{13}$\\
\hline
\hline $[(n, n-4i+4, 2; 1, 2i+1{)]}_{q}$, $n=(q^{2}+1)/2$, $2\leq i\leq (q-1)/2$\\
\hline ${[(25, 21, 2; 1, 5)]}_{7}$\\
\hline ${[(25, 17, 2; 1, 7)]}_{7}$\\
\hline ${[(61, 57, 2; 1, 5)]}_{11}$\\
\hline ${[(61, 53, 2; 1, 7)]}_{11}$\\
\hline ${[(61, 49, 2; 1, 9)]}_{11}$\\
\hline ${[(61, 45, 2; 1, 11)]}_{11}$\\
\hline ${[(145, 141, 2; 1, 5)]}_{17}$\\
\hline ${[(145, 137, 2; 1, 7)]}_{17}$\\
\hline ${[(145, 133, 2; 1, 9)]}_{17}$\\
\hline ${[(145, 129, 2; 1, 11)]}_{17}$\\
\hline ${[(145, 125, 2; 1, 13)]}_{17}$\\
\hline ${[(145, 121, 2; 1, 15)]}_{17}$\\
\hline ${[(145, 117, 2; 1, 17)]}_{17}$\\
\hline
\end{tabular}
\end{center}
\end{table}


\section{Summary}\label{VII}
In this paper we have constructed new families of classical and
quantum MDS-convolutional codes derived from negacyclic codes. All
the constructions presented here are performed algebraically and
not by exhaustively computational search. The results obtained in
this paper show that the class of negacyclic codes is also a good
source in the search for optimal codes.

\section*{Acknowledgment}
This research has been partially supported by the Brazilian
Agencies CAPES and CNPq.

\textbf{Giuliano G. La Guardia received the M.S. degree in pure
mathematics in 1998 and the Ph.D. degree in electrical engineering
in 2008, both from the State University of Campinas (UNICAMP),
Brazil. Since 1999, he has been with the Department of Mathematics
and Statistics, State University of Ponta Grossa, where he is an
Associate Professor. His research areas include theory of
classical and quantum codes, matroid theory, and error analysis.}

\end{document}